\documentclass[runningheads]{llncs}
 
\usepackage{graphicx}
\usepackage{verbatim}
\usepackage{amsmath,amssymb}
\usepackage{enumerate}
\usepackage{amsfonts}
  
  \usepackage[usenames,dvipsnames]{pstricks}
 \usepackage{epsfig}
 \usepackage{pst-grad} 
 \usepackage{pst-plot} 

\begin{document}

\titlerunning{Monotone Grid Drawings of Planar Graphs}
\authorrunning{Hossain and Rahman}
\newtheorem{fact}{Fact}{\bf}{\it}

\mainmatter
\title{Monotone Grid Drawings of Planar Graphs}

\author{
Md.\ Iqbal\ Hossain and 
Md.\ Saidur Rahman
} 

\institute{
Graph Drawing and Information Visualization Laboratory,\\
Department of Computer Science and Engineering,\\
Bangladesh University of Engineering and Technology (BUET),\\
Dhaka-1000, Bangladesh.\\
\email{mdiqbalhossain@cse.buet.ac.bd, saidurrahman@cse.buet.ac.bd}
}

\maketitle

\begin{abstract} A monotone drawing of a planar graph $G$ is a planar straight-line  drawing of $G$ where a monotone path exists
 between every pair of vertices of $G$ in some direction. Recently monotone
	drawings of planar graphs have been proposed as a new standard for visualizing graphs. A monotone drawing of a planar graph is a monotone grid drawing if every vertex in the drawing is drawn on a grid point. In this paper we
	study monotone grid drawings of planar graphs in a
variable embedding setting. We show that every connected planar graph of $n$ vertices has a  
	 monotone grid drawing on a grid of size $O(n)\times O(n^2)$, and such a drawing can be found in $O(n)$ time.
\end{abstract}

\section{Introduction}

\label{Introduction} 
A {\it straight-line drawing} of a planar graph $G$ is a drawing of $G$ in which each vertex is drawn as a point and each edge   is drawn as a straight-line segment without any edge crossing.
A path $P$ in a straight-line drawing of a planar graph is {\it monotone} if there
exists a line $l$ such that the orthogonal projections of the vertices of $P$ on $l$
appear along $l$ in the order induced by $P$. A straight-line drawing $\Gamma$ of a  planar graph $G$
is a {\it monotone drawing} of $G$ if $\Gamma$  contains at least one monotone path between every pair of vertices. In the drawing of a graph in Fig.~\ref{fig:mononotmono}, the path between the vertices $s$ and $t$ drawn as a thick line is a monotone path with respect to the direction $d$, whereas no monotone path exists with respect to any direction between the vertices $s'$ and $t'$. We call a monotone drawing of a planar graph a {\it monotone grid drawing} if every vertex is drawn on a grid point.


\begin{figure}
\centering
\includegraphics[scale=0.3]{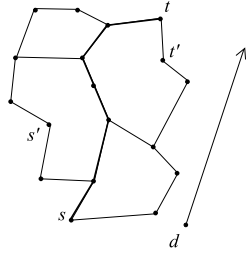}
\caption{The path between vertices $s$ and $t$ (as shown by thick line) is monotone with respect to direction $d$.}
\label{fig:mononotmono}
\end{figure}

 Monotone drawings of graphs are well motivated by human subject experiments by Huang {\it et al.}~\cite{Huangx}, who showed that the ``geodesic tendency" (paths following a given direction) is important in comprehending the underlying graph.
{\it Upward drawings}~\cite{DB88,MAHMSR07,BattistaETT99,NR04}  are related to monotone drawings where every directed path is monotone with respect to the vertical line, while in a monotone drawing each monotone path, in general, is monotone with respect to a different line.
Arkin {\it et al.}~\cite{Ar89} showed that any strictly convex drawing of a planar graph is monotone and they gave an $O(n\log n)$ time algorithm  for finding such a path between a pair of vertices in a strictly convex drawing of a planar graph of $n$ vertices. Angelini {\it et al.}~\cite{PAT12} showed that every biconnected planar graph of $n$ vertices has a
  monotone drawing in real coordinate space.
They also showed that every tree of $n$ vertices admits a  monotone grid drawing on a grid of $O(n)\times O(n^2)$ or  $O(n^{1.6})\times O(n^{1.6})$
area. 
It is known that every outerplane graph of $n$ vertices admits a monotone grid drawing on a grid of area $O(n)\times O(n^2)$~\cite{Ang11}. Recently, Hossain  and Rahman~\cite{IQ12} showed that every series-parallel graph of $n$ vertices admits a monotone grid drawing on an  $O(n)\times O(n^2)$ grid, and such a drawing can be found in  $O(n\log n)$ time. 
 It is also known that not every plane graph (with fixed embedding) admits a monotone
drawing~\cite{PAT12}.

In this paper we investigate whether every connected planar graph has a monotone
drawing and what are the area requirements for  such a drawing on a grid. We show that  every connected planar graph of $n$ vertices has a   monotone grid drawing on an $O(n)\times O(n^2)$ grid, and such a drawing can be computed in  $O(n)$  time.
As a byproduct, we introduce a spanning tree of a plane graph with some interesting properties. Such a spanning tree may find applications in other areas of graph algorithms as well.

 \begin{figure}
\centering
\includegraphics[scale=0.6]{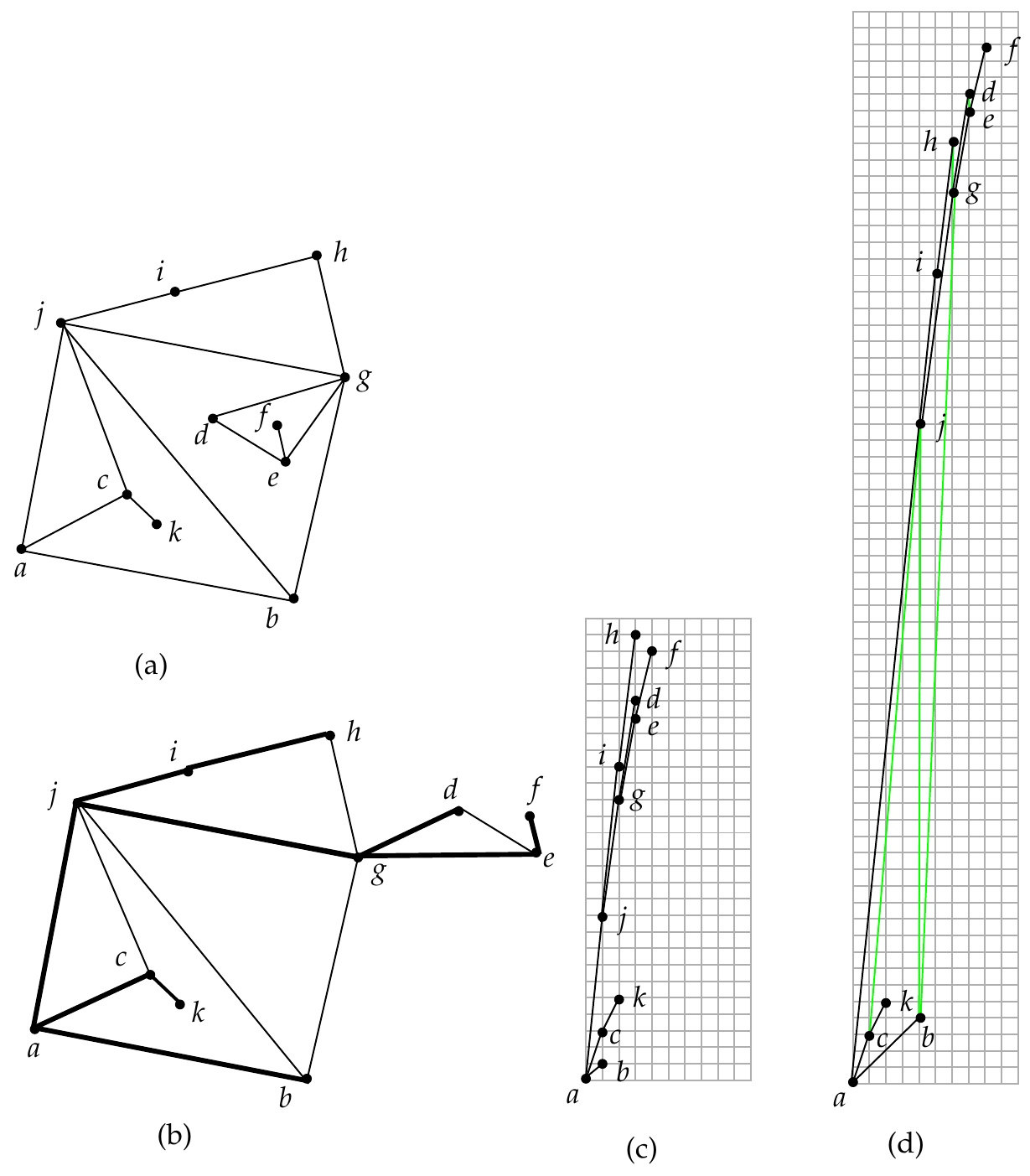}
\caption{ Illustration for an outline of our algorithm}
\label{fig:Overview}
\end{figure}

We now give an outline of  our algorithm for constructing a monotone grid drawing of a planar graph $G$. We first construct a ``good spanning tree'' $T$ of $G$ and find a monotone drawing of $T$ by the method given in~\cite{PAT12}. We then draw each non-tree edge by a straight-line segment by shifting the drawing of some subtree of $T$, if necessary. Figure~\ref{fig:Overview} illustrates the steps of our algorithm. 
The input planar graph $G$ is shown in Fig~\ref{fig:Overview}(a). We first find a planar embedding of $G$ containing a good spanning tree as illustrated in Fig.~\ref{fig:Overview}(b), where the edges of the spanning tree are drawn by thick lines.
 We then find a monotone drawing of $T$ on $O(n)\times O(n^2)$ grid using the algorithm  in~\cite{PAT12} as illustrated in Fig.~\ref{fig:Overview}(c). Finally we elongate the drawing of some edges and draw the non-tree edges of $G$ using straight-line segments as illustrated in Fig.~\ref{fig:Overview}(d).

The rest of the paper is organized as follows.
Section~\ref{preliminaries} describes some of the definitions that we have used in our paper.
 Section~\ref{monotonesp} deals with  monotone drawings of connected planar graphs.
 Finally, Section~\ref{conclusion} concludes the paper with discussions.

\section{Preliminaries}
\label{preliminaries} 
 In this section we give some definitions and present a known result. For the graph theoretic terminologies not given here, see~\cite{NR04}.

Let $G = (V,E)$ be a connected graph with vertex set $V$ and edge set $E$. A {\it subgraph} of $G$ is a graph $G'=(V',E')$ such that $V' \subseteq V$  and $E' \subseteq E$. The degree of a vertex $v$ in $G$ is denoted by $d(v)$. We denote an edge joining vertices $u$ and $v$ of $G$ by $(u,v)$. 
A pair $\{u, v\}$ of vertices in $G$ is a split pair if there exist two subgraphs $G_1 =
(V_1,E_1)$ and $G_2 = (V_2,E_2)$ satisfying the following two conditions:
1. $V = V_1 \cup V_2, V_1 \cap V_2 = \{u, v\}$; and
2. $E = E_1 \cup  E_2, E_1 \cap E_2 = \emptyset, |E_1| \geq 1, |E_2| \geq 1$.
Thus every pair of adjacent vertices is a split pair. A {\it $\{u,v\}$-split component} of  a split
pair ${u, v}$ in $G$ is either an edge $(u, v)$ or a maximal connected subgraph $H$ of
$G$ such that $\{u, v\}$ is not a split pair of $H$.
If $v$ is a vertex in $G$, then $G-v$ is the subgraph of $G$ obtained by deleting the vertex $v$ and all the edges incident to $v$. Similarly, if $e$ is an edge of $G$, then $G-e$ is a subgraph of $G$ obtained by deleting the edge $e$.
Let $v$ be a cut-vertex in a connected graph $G$. We call a subgraph $H$ of $G$ a {\it $v$-component} if $H$ consists of a connected component $H'$ of $G-v$ and all edges joining $v$ to the vertices of $H'$.

Let $G_1=(V_1,E_1)$ and $G_2=(V_2,E_2)$ be two graphs. The {\it union } of $G_1$ and $G_2$, denoted by $G_1 \cup G_2$, is a graph $G_3=(V_3,E_3)$ such that $V_3=(V_1 \cup V_2)$ and  $E_3=(E_1 \cup E_2)$.

Let $G=(V,E)$ be a graph and $T=(V,E')$ be a spanning tree of $G$. An edge $e\in E$ is called a {\it tree edge} if $e\in E'$ otherwise $e$ is said to be a {\it non-tree edge}. Let $e$ be a non-tree edge with respect to $G$ and $T$. Then by $T\cup e$ we denote the subgraph $G'$ of $G$ obtained adding edge $e$ to $T$.   The graph $G'=T \cup e$ always has a single cycle $C$ and  we call $C$ the {\it cycle induced by the non-tree edge} $e$. If $X$ is a set of edges of $G$ and $T$ is a spanning tree of $G$ then $T\cup X$ denotes the graph obtained by adding the edges in $X$ to $T$ and replacing each multi-edge by a single edge. 
Let $T$ be a rooted tree and let $u$ be a vertex of $T$. Then by $T_u$ we denote the subtree of $T$ rooted at $u$. By $T-T_u$ we denote the tree obtained from $T$ by deleting the subtree $T_u$.

A graph is {\it planar} if it can be embedded in the plane without edge intersections
except at the vertices where the edges are incident. A {\it plane graph } is a planar
graph with a fixed planar embedding. A plane graph divides the plane into some
connected regions called {\it faces}. The unbounded region is called the {\it outer face}
and each of  the other faces is called an {\it inner face}. 
Let $G$ be a plane graph. The boundary of the outer face of $G$ is called the {\it outer boundary} of $G$. We call a simple cycle induced by the outer boundary of $G$ an {\it outer cycle} of $G$.  We call a vertex $v$ of $G$ an {\it outer vertex} of $G$ if $v$ is on the outer boundary of $G$, otherwise $v$ is an {\it inner vertex} of $G$.



Let $p$ be a point in the plane and $l$ be a half-line with an end at $p$. The slope of $l$,
denoted by $slope(l)$, is the angle spanned by a counterclockwise rotation that
brings a horizontal half-line started at $p$ and directed towards increasing $x$-
coordinates to coincide with $l$.
Let $\Gamma$ be a drawing of a graph $G$ and let $(u, v)$ be an edge of $G$. We denote the direction of a half-line by $d(u,v)$ which is started at $u$  and passed through $v$. The  direction of a drawing of an edge  $e$ is denoted by $d(e)$ and the slope of the drawing of $e$ is denoted by $slope(e)$.

Let $G$ be a planar graph and $\Gamma$ be a straight-line drawing of $G$. A path $u = u_1, \ldots , u_k = v$
between vertices $u$ and $v$ in $G$ is denoted by $P(u, v)$. The drawing of the path $P(u, v)$  in $\Gamma$ is {\it monotone} with respect to a direction
$d$ if the orthogonal projections of vertices $u_1,\ldots , u_k$ on $d$ appear
in the same order as the vertices appear on the path. The drawing $\Gamma$ is a {\it monotone drawing} of $G$
if there exists a direction $d$ for every pair of vertices $u$ and $v$ such that $P(u, v)$ is monotone with respect to $d$. A monotone drawing is a {\it monotone grid drawing} if every vertex is drawn on a grid point.  The following lemma is known from~\cite{PAT12}.



\begin{lemma}
 Let $T$ be a tree of $n$ vertices. Then $T$ admits a monotone grid drawing on a grid of area $O(n) \times O(n^2)$, and such a drawing can be found in $O(n)$ time.
\label{lemmatreedrawing}
\end{lemma}
In this paper we use a modified version of the algorithm for monotone grid drawing of a tree in~\cite{PAT12}, which we call  Algorithm {\bf Draw-Monotone-Tree} throughout this paper.
   Algorithm {\bf Draw-Monotone-Tree} first assigns a slope to each vertex of a planar embedded tree then obtains a slope-disjoint drawing of the tree which is monotone. 
A brief description of the algorithm is given in the rest of this section.
Let $T$ be an embedded rooted tree of $n$ vertices. (Note that in~\cite{PAT12} $T$  is not embedded, but here we use $T$ as an embedded tree for the sake of our algorithm.)
Let $S= \{s_1, s_2,\ldots, s_{n-1}\}= \{1/1, 2/1, 3/1,\ldots, (n-1)/1\}$ be an ordered set of $n$ slopes in increasing order,  where each slope is represented by the ratio $y/x$.  
Let $v_1, v_2, \ldots, v_n$ be an ordering of vertices in $T$ in a counterclockwise postorder traversal. (In a counterclockwise postorder traversal of a rooted ordered tree, subtrees rooted at the children of the root are recursively traversed in counterclockwise order and then the root is visited.) Then we assign the slope $s_i$ to vertex $v_i$ $(i\ne n)$.
Let  $u_1, u_2, \ldots,u_k$ be the children of $v$ in $T$. Then the subtree $T_{u_i}$ gets $|T_{u_i}| $ consecutive elements of $S$ from the $(1+\sum^{i-1}_{j=1}   |T_{u_j}|)$-th to the $(\sum^{i}_{j=1}|T_{u_j}|)$-th. Let $v'$ be the parent of $v$. If $v$ is not the root of $T$  then the drawing of the edge $e=(v',v)$ will be a straight-line with slope $s_i$.

We now describe how to find a monotone grid drawing of $T$ using the slope assigned to each vertex of $T$. We first draw the root vertex $r$ at $(0,0)$, and then use a counter clockwise preorder traversal for drawing each vertex of $T$. (In a counterclockwise preorder traversal of a rooted ordered tree, first the root is visited and then the subtrees  rooted at the children   of the root are visited recursively in counterclockwise order.)  We fix the position of a vertex $u$ when we traverse $u$. Note that when we traverse $u$, the position of the parent $p(u)$ has already been fixed. Let $(p_x (u), p_y (u))$ be the position of $p(u)$. Then we place $u$ at grid point $(p_x (u) + x_b , p_y (u) + y_b )$, where $s_b = y_b /x_b$. Figure~\ref{fig:treeexample}(b) illustrates a monotone grid drawing of the tree as shown in the Fig.~\ref{fig:treeexample}(a).
Algorithm {\bf Draw-Monotone-Tree} computes a slope-disjoint monotone drawing of a tree on an  $O(n) \times O(n^2)$ grid in linear time~\cite{PAT12}. 


\begin{figure}
\centering
\includegraphics[scale=0.8]{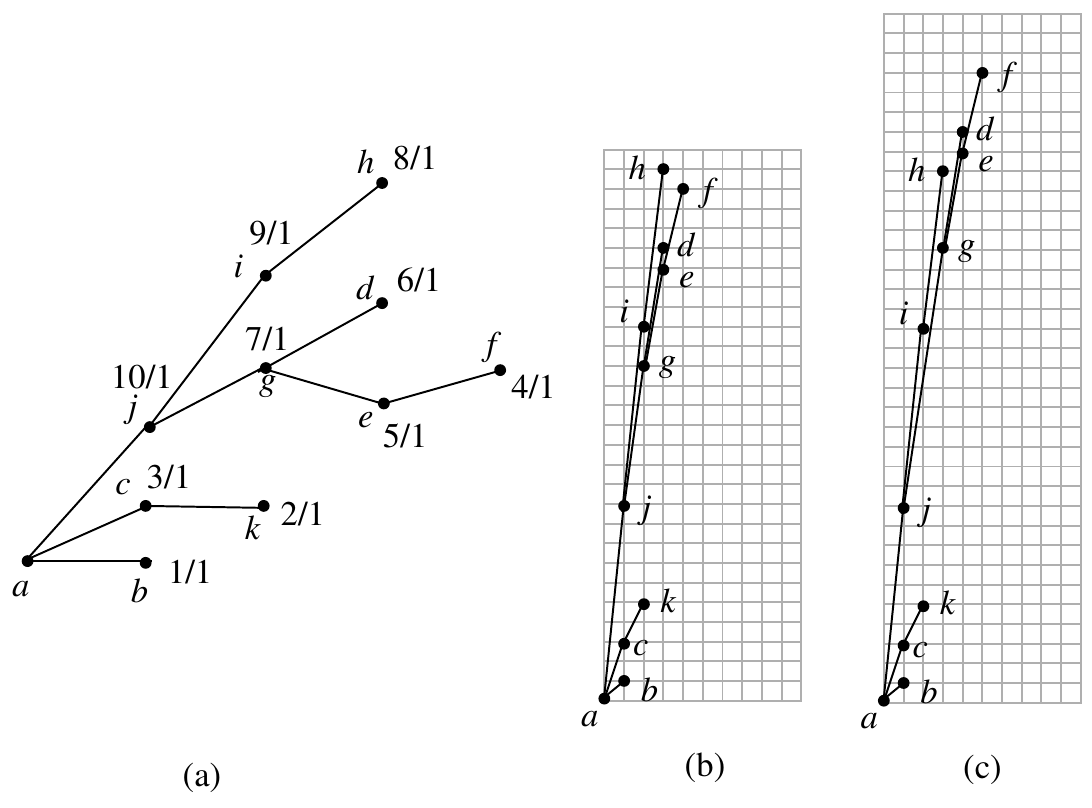}
\caption{(a) A tree $T$ with assigned slope to each vertex, (b) a monotone drawing $\Gamma$ of $T$ and (c) a monotone drawing $\Gamma'$ of $T$ with elongation of the edge $(j,g)$.}
\label{fig:treeexample}
\end{figure}


\section{Monotone Grid Drawings}
\label{monotonesp}

 In this section we show that every connected planar graph of $n$ vertices has a monotone grid drawing on an $O(n)\times O(n^2)$ grid.
 
 Let $G$ be a planar graph and let $G_\phi$ be a plane embedding of $G$. Let $T$ be an ordered rooted spanning tree of $G_\phi$ such that the root $r$ of $T$ is an outer vertex $G_\phi$, and the ordering of the children of each vertex $v$ in $T$ is consistent with the ordering  of the neighbors of $v$ in $G_\phi$. 
Let $P(r,v)=(r=u_1), u_2, \ldots, (v=u_k)$ be the path in $T$ from the root $r$ to a vertex $v\ne r$.
The path $P(r,v)$ divides the children of $u_i$, $(1\le i < k)$, except $u_{i+1}$, into two groups; the left group $L$ and the right group $R$. A child $x$ of $u_i$ is in group $L$  and denoted by $u_{i}^L$ if the edge $(u_i,x)$ appears before the edge $(u_i, u_{i+1})$ in clockwise ordering of the edges incident to $u_i$ when the ordering is started from the edge $(u_i,u_{i+1})$, as illustrated in the Fig.~\ref{fig:gridspliting}(a). Similarly, a child $x$ of $u_i$ is in the group $R$ and denoted by $u_{i}^R$ if the edge $(u_i,x)$ appears after the edge $(u_i, u_{i+1})$ in clockwise order of the edges incident to $u_i$ when the ordering is started from the edge $(u_i,u_{i+1})$.
 We call $T$  a {\it good spanning} tree of $G_\phi$ if every vertex $v$ $(v\ne r)$ of $G$ satisfies the following two conditions with respect to $P(r,v)$.
\begin{enumerate}
\item [(Cond1)] $G$ does not have a non-tree edge $(v,u_i)$, $i<k$; and

\item [(Cond2)]
the edges of $G$ incident to the vertex $v$ excluding $(u_{k-1},v)$ can be partitioned into three disjoint (possibly empty) sets $X_v$, $Y_v$ and $Z_v$ satisfying the following conditions (a)-(c) (see Fig.~\ref{fig:gridspliting}(b)): 
\begin{enumerate}
\item [(a)] Each of $X_v$ and $Z_v$ is a set of consecutive non-tree edges and $Y_v$ is a set of consecutive tree edges.
\item [(b)] Edges  of set $X_v$, $Y_v$ and $Z_v$ appear clockwise in this order from the edge $(u_{k-1}, v)$.
\item [(c)] For each edge $(v,v')\in X_v$, $v'$ is contained in $T_{u_i^L}$, $i<k$, and for each edge $(v,v')\in Z_v$, $v'$ is contained in $T_{u_i^R}$, $i<k$.

\end{enumerate}  


\end{enumerate}  
\begin{figure}
\centering
\includegraphics[scale=0.7]{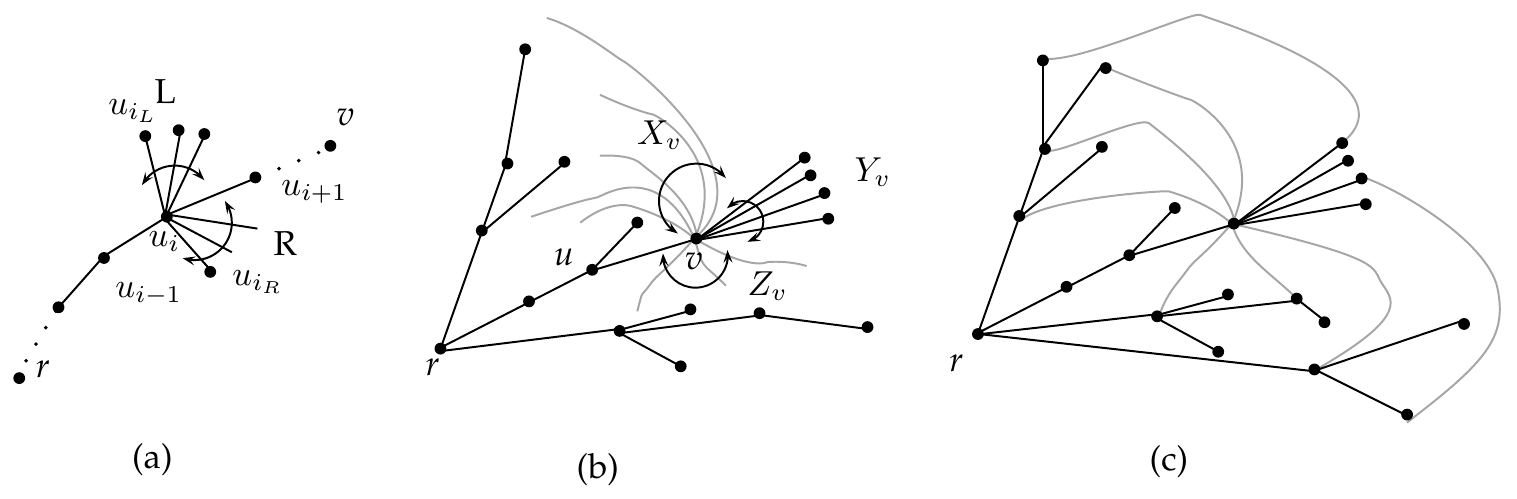}
\caption{(a) An illustration for $P(r,v)$, $L$ and $R$ groups,   (b) an illustration for $X_v$, $Y_v$ and $Z_v$ sets of edges, and (c) an illustration for a good spanning tree $T$ on $G_\phi$ where bold edges are tree edges.}
\label{fig:gridspliting}
\end{figure}
Figure~\ref{fig:gridspliting}(c) illustrates a good spanning tree $T$ in a plane graph. 
The following two lemmas are based on the properties of a good spanning tree and on Lemma~\ref{lemmatreedrawing}.
\begin{lemma}
\label{lemma2}

Let $T$ be a good spanning tree of $G_\phi$. Let $\Gamma$ be a monotone drawing of $T$. Let $\Gamma'$ be a straight-line drawing of $T$ obtained from $\Gamma$ by elongation of the drawing of an edge $e$ of $T$ preserving the slope of $e$. Then $\Gamma'$ is also a monotone drawing. (See Fig.~\ref{fig:treeexample}(c), where the edge $(j,g)$ is elongated.)

\end{lemma}

   \begin{proof}
Let $e=(u,v)$ be an edge of $T$ where $u$ is the parent of $v$ in $T$. The elongation of the drawing of $e$ does not change $slope(e)$ and the drawing of $T_v$ in $\Gamma'$ is shifted outwards preserving its drawing in $\Gamma$. The drawing of $T- T_v$ is remained same in $\Gamma'$. Since the elongation does not change the slope of the drawing of any edge, the new drawing $\Gamma'$ preserves monotone drawing of $T$. 
 
   \end{proof}

\begin{lemma}
\label{lemma3}

Let $T$ be a good spanning tree of $G_\phi$. 
 Let $v$ be a vertex in $G$, and let $u$ be the parent of $v$ in $T$. Let $X\subseteq X_v$ and $Z\subseteq Z_v$. 
Assume that $\Gamma$ is a the  monotone drawing of $T\cup X \cup Z$ where a monotone path exists between every pair of vertices in the drawing of $T$ in $\Gamma$. 
If a straight-line drawing  $\Gamma'$ of $T\cup X \cup Z$ is obtained from $\Gamma$ by elongation of the drawing of the edge $(u,v)$, then $\Gamma'$ is a monotone drawing of  $T\cup X \cup Z$ where a monotone path exists between every pair of vertices in the drawing of $T$ in $\Gamma'$. 

\end{lemma}


  \begin{proof}

Let $r$ be the root of $T$. Let $m$  be the slope assigned to the vertex $v$ in $T$.


Let $M_X$ and $M_Z$ be the sets of slopes assigned to $T_{u_i^L}$ and $T_{u_i^R}$, respectively.
According to assignment of slopes, for any $m_x\in M_X$ and $m_z\in M_Z$ the relation $m_x>m>m_z$ holds.


Since each vertex in $X$ and $Z$ are visible from the vertex $v$ in $\Gamma$ and  $m_x<m<m_z$,
 $v$ must be visible from  each vertex in $X$ and $Z$ even after elongation of edge $(u,v)$ without changing the slope of $(u,v)$. 
Note that the elongation only changes the slopes of the drawings of non-tree edges  in $X$ and $Z$. The drawing of $T_v$ is shifted outwards preserving the slopes of the   edges in $T_v$ and the drawing of $T - T_v$ is remained same. Let $\Gamma'$ be the new drawing of $T\cup X \cup Z$. Then obviously the edges in $X$ and in $Z$ does not produce any edge crossing in $\Gamma'$. By Lemma~\ref{lemma2}, the elongation of the edge $(u,v)$ does not break the monotone property in the drawing of $T$ in $\Gamma'$. Thus a monotone path exists between every pair of vertices in the drawing of $T$ in $\Gamma'$.

  \end{proof}

We  now have the following lemma on monotone grid drawings of a plane graph with a good spanning tree.

\begin{lemma}
\label{theo1}
Let $G$ be a planar graph of $n$ vertices and let $G_\phi$ be a plane embedding of $G$. Assume that $G_\phi$ has a good spanning tree $T$.
Then $G_\phi$ admits a   monotone grid drawing on an $O(n)\times O(n^2)$ grid. 
\end{lemma}
  \begin{proof}

Let $T$ be a good spanning tree of $G_\phi$. We prove the claim by induction on the number $z$ of non-tree edges in $G$ with respect to $T$. 
Algorithm {\bf Draw-Monotone-Tree} uses a counterclockwise postorder traversal for finding a vertex ordering in the tree and assigns slope to each vertex of the tree using that ordering. 
Note that the ordering of the vertices is fixed once the child of $r$ that has to be visited first is fixed. Let $T$ be a good spanning tree of $G_\phi$ and let $r$  be the root of $T$. We take a child $s$ of $r$ as the first child to be visited in counterclockwise postorder traversal such that $s$ is an outer vertex of $G_\phi$ and if $s$ is on an outer cycle $C$ of $G_\phi$ then $s$ is the counterclockwise neighbor of $r$ on $C$. We call the edge $(r,s)$ the {\it reference edge} of $T$. (Later in Corollary 1 we show that such a reference edge always exists.)  
By induction on the number $z$ of non-tree edges  of  $G_\phi$, we now prove the claim that  $G_\phi$  admits a monotone grid drawing on     an $O(n) \times O(n^2) $ grid and a monotone path exists between every pair of vertices of  $G_\phi$  through the edges of $T$ in the drawing.


We first assume that $z=0$. In this case $G_\phi=T$. 
We then find monotone drawing $\Gamma$ of $T$ on an $O(n) \times O(n^2) $ grid using Algorithm {\bf Draw-Monotone-Tree} taking the reference edge $(r,s)$ as the starting edge for traversals (counterclockwise postorder traversal for slope assignment and counterclockwise preorder traversal  for drawing vertices). Since  $G_\phi=T$, $\Gamma$ is a monotone drawing of $G_\phi$. That is, a monotone path exists between every pair of vertices of $T$ in $\Gamma$.


We thus assume that $z>0$ and the claim holds for any plane graph $G_\phi$ with number of non-tree edges $z'$, where $z'<z$.

Let $G_\phi$ have $z$ non-tree edges with respect to $T$ and let  $e=(u,v)$ be a non-tree edge of the outer boundary of $G_\phi$.

Let $m_u$ and $m_v$ be the  slopes assigned to the vertices $u$ and $v$, respectively in $T$ by Algorithm {\bf Draw-Monotone-Tree}. Without loss of generality let us assume $m_u>m_v$. Let $w$ be the common ancestor of $u$ and $v$ in $T$ and let $u'$ and $v'$ be the parents of $u$ an $v$ in $T$. According to (Cond1) $u$ does not lie on the path $P(v,r)$ and $v$ does not lie on the path $P(u,r)$. Let $C=\{P(u,w)\cup P(v,w) \cup (u,v)\}$ and $G_\phi'$= $\{P(r,w) \cup C\}$. Clearly $G_\phi' - (u,v)$ has $l_1$ non-tree edges where $l_1<z$.  By induction hypothesis, $G_\phi'-(u,v)$ has a straight-line monotone drawing on $O(n) \times O(n^2) $ grid where the edges in $T$ are drawn with the slope assigned to them and a monotone path exists between every pair of vertices through the edges in $T$.  Let $\Gamma'$ be the drawing of $G_\phi' - (u,v)$ in  $\Gamma$.  Let $p_x$ be the largest $x$-coordinate used for the drawing of  $\Gamma'$. We now shift the drawing of $T_u$ and $T_v$ on the line $x=p_x+1$ by preserving slopes of the drawings of the edges $(u',u)$ and $(v', v)$. Since the slopes are integer numbers, it guarantees that all vertices remain on grid points after the shifting operation.
 According to Lemma~\ref{lemma2} elongations of $(u',u)$ and $(v', v)$ do not produce any edge crossing in the drawing.  According to (Cond2) in the good spanning tree $T$, $e$ belongs to set $Z_u$ and set $X_v$. Then no tree edge indecent to $u$ exists between the edge $(u',u)$ and $(u,v)$ in counterclockwise from the edge $(u',u)$, and no tree edge indecent to $v$ exists between the edge $(v',v)$ and $(v,u)$ in clockwise from the edge $(v',v)$. (Remember that we have used counterclockwise postorder traversal starting from a reference edge for ordering the vertices in algorithm {\bf Draw-Monotone-Tree}.)  Hence we can draw the edge $e$ on the line $x=p_x+1$ by a straight-line segment without any edge crossings.  In worst case,  $y$-coordinate can be at most $O(n^2)$ + $O(n^2)$. Hence the drawing takes a grid of  size $O(n) \times O(n^2)$.
 
 \end{proof}


We now prove that every connected planar graph $G$ has an embedding $G_\phi$ where a good spanning tree $T$ of $G_\phi$ exists. We give a constructive proof for our claim. Before giving our formal proof we give an outline of our construction using an illustrative example in Fig.~\ref{fig:completebfs}. We take an arbitrary plane embedding    $G_\gamma$  of $G$ and start breath-first-search (BFS) from an arbitrary outer vertex $v$ of $G_\gamma$ and regard $r$ as the root of our desired spanning tree. In Fig.~\ref{fig:completebfs}(a)  BFS is started from vertex $a$, and vertex $b,c$ and $d$ are visited from $a$ in this order by BFS, as illustrated in Fig.~\ref{fig:completebfs}(b).  Next we visit $e$ from $b$ by BFS, as illustrated in Fig.~\ref{fig:completebfs}(c). When we visit a new vertex $x$ then we check whether  there is an edge $(x,y)$ such that $y$ is already visited and there is an   $(x,y)$-split component or an $x$-component or a $y$-component inside the cycle induced by the edge $(x,y)$ which does not contain the root $r$. The $\{e,d\}$-split component $H_1$ induced by the vertices $\{d,h,i,j,e\}$ is such a split component in Fig.~\ref{fig:completebfs}(c) and the subgraph $H_2$ induced by vertices ${d,m,n}$ is such a $y$-component for $y=d$ which are inside the cycle induced by the edge $(e,d)$. We move the subgraphs $H_1$ and $H_2$ out of the cycle induced by the non-tree edge $(e,d)$, as illustrated in Fig.~\ref{fig:completebfs}(d). Since $(b,e)$ is a tree edge and $(e,d)$ is a non-tree edge, according to definition of a good spanning tree, the edges $(e,f)$ and $(e,k)$ must be non-tree edges. Similarly since $(a,d)$ is a tree edge and $(e,d)$ a non-tree edge, then the edge $(d,l)$ must be non-tree edge. We mark  $(e,f)$, $(e,k)$ and $(d,l)$ non-tree edges as shown in the Fig.~\ref{fig:completebfs}(e). We then visit vertices $f,l,m,n$ and $g$, as illustrated in Fig.~\ref{fig:completebfs}(f). When we visit $k$,  we find a $k$-component $H$ induced by vertices $\{k,p,o\}$ and we move $H$ out of the cycle induced by $(e,k)$ as shown in Fig.~\ref{fig:completebfs}(g). Finally, at the end of BFS we find an embedding  $G_\phi$ of $G$ and a good spanning tree $T$ as illustrated in  Fig.~\ref{fig:completebfs}(h), where the good spanning tree $T$ is shown by  solid edges, and non-tree edges are shown by dashed edges. 

 \begin{figure}
\centering
\includegraphics[scale=0.38]{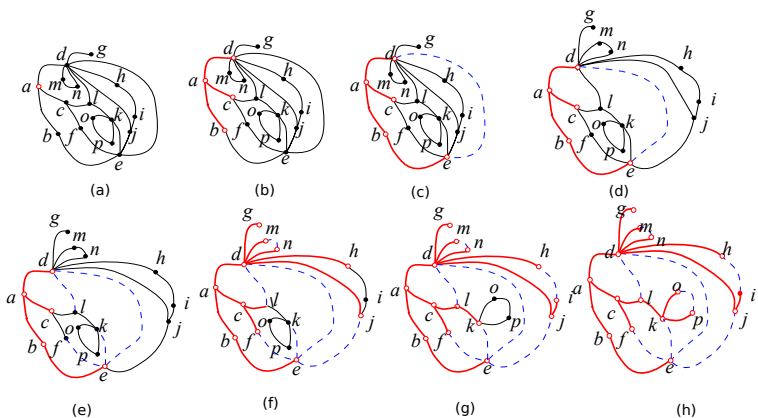}
\caption{Illustration for an outline of construction of a good spanning tree $T$. White vertices are visited vertices. Black vertices are not visited. Solid edges are tree edges. Dashed edges are non-tree edges.}
\label{fig:completebfs}
\end{figure}

We now formally prove our claim as in the following lemma.
\begin{lemma}
Let $G$ be a planar graph of $n$ vertices. Then $G$ has a plane embedding $G_\phi$ that contains a good spanning tree.
 \label{lemmatree}
\end{lemma}

\begin{proof}
 
We give a constructive proof.
Let $G_\gamma$ be any arbitrary embedding of $G$.
We first   mark an arbitrary outer vertex $r$ of $G_\gamma$ visited, and start clockwise BFS from $r$. The vertex $r$ will be the root of the BFS tree $T$.

Note that after visiting each vertex by BFS, the embedding of $G$ may be changed by our algorithm. Let $G_\gamma^i$  be the embedding of $G$ after visiting the $i$th vertex by BFS. Then $G_\gamma^1= G_\gamma$ since we do not change the embedding after visiting $r$. Let $T^i$ be the BFS tree after visiting  the $i$th vertex. Then $T^1$ contains the single vertex $r$. In counterclockwise BFS, we first visit a neighbor $s$ of $r$ such  
that $s$ is an outer vertex of $G_\gamma$ and if $s$ is on an outer cycle $C$ of $G_\gamma$ then $s$ is the counterclockwise neighbor of $r$ on $C$. We call the edge $(r,s)$ the {\it BFS-Start edge}. 


We now assume that  vertices $w_1(=r),w_2,\ldots,w_{j-1}$ $(j-1 < n)$ are visited by BFS and  we are visiting $w_j$ form $w'$, that is,  $w'$ is the parent of $w_{j}$ in $T^j$. We mark $w_j$ as visited and mark $(w', w_j)$ as a tree edge. If there is no edge $e=(w_j, v)$ such that $v \in  V(T^{j-1})$ and $v\ne w'$, then we proceed for the next vertex $w_{j+1}$. Otherwise, an edge $e=(w_j, v)$ exists such that $v \in  V(T^{j-1})$ and $v\ne w'$. In such a case we mark $e$ as a non-tree edge, and change the embedding of $G_\gamma^{j-1}$ to get $G_\gamma^{j}$, if necessary, as follows. 

We set $x=w_j$ and $y=v$ if $v$ comes earlier in the counterclockwise postorder traversal of $T^j$ started from  
 the BFS-Start edge; otherwise, we set  $x=v$ and $y=w_j$.
Let $C$ be the cycle induced by the non-tree edge $e=(x,y)$. Let $G_j(C)$ be the plane subgraph of $G_\gamma^{j}$ inside $C$ (including $C$).
We  check whether there is any $(x,y)$-split component or $x$-component or $y$-component in $G_j(C)$.  If there is a $(y,x)$-split component or a $x$-component or a $y$-component $H$ in $G_j(C)$ such that  $r\notin V(H)$. We move $H$ out of the cycle $C$ and obtain embedding $G_\gamma^j$. In Fig.~\ref{fig:sperationpair}(a) $I_1,\ldots,I_k$ are $(y,x)$-split components, $J_1,\ldots,J_l$ are $x$-components and  $K_1,\ldots,K_m$ are $y$-components, are move out of $C$ in  Fig.~\ref{fig:sperationpair}(b).

\begin{figure}
\centering
\includegraphics[scale=0.5]{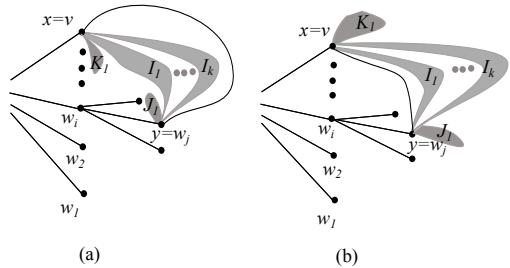}
\caption{ Illustration for $(x,y)$-split components, $x$-components and $y$-components.}
\label{fig:sperationpair}
\end{figure}

 One can observe that unvisited vertices inside cycle $C$ can be accessed from a visited vertex other than $x$ and $y$  in a later BFS step.
We thus mark some edges  incident to $x$ and $y$  as non-tree edges for maintaining the properties of a good spanning tree for $T^j$, as follows.
Let $y_p$ and $x_p$ be the parent of $y$ and $x$ in $T^j$, respectively. 
Let  $E_z=\{e_1,e_2,\ldots, e_k\}$ be the set of edges which are incident to the vertex $x$, and between the edges $(x_p, x)$ and  $(x,y)$ in counterclockwise order starting from the edge $(x_p, x)$ in $G_\gamma^j$, as shown in   Fig.~\ref{fig:blackedge}.  We mark the edges in $E_x$ as non-tree edges.  Note that the edges set  $E_z \cup \{(x,y)\}$ will be the set $Z_x$  with respect to the vertex $x$ in the final good spanning tree. Let $E_x=\{e_1,e_2,\ldots,e_l\}$ be the edges which are incident to the vertex $y$ between the edges $(y_p, y)$ and  $(y,x)$ in clockwise order started from the edge  $(y_p, y)$ in $G_\gamma^j$ as shown in   Fig.~\ref{fig:blackedge}. We also mark the edges in $E_x$ as non-tree edges.  The edge set  $E_x \cup \{(y,x)\}$ will be the set $X_y$ with respect to the vertex $y$ in the final spanning tree.

 \begin{figure}
\centering
\includegraphics[scale=0.5]{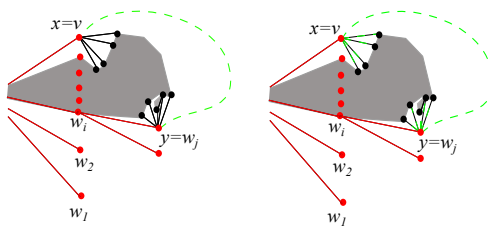}
\caption{ Marking non-tree edges from $x$ and $y$.}
\label{fig:blackedge}
\end{figure}
  
Finally we get $G_\gamma^n$ and $T^n$ after visiting $n$th vertex. 
We now show that $T^n$ is a good spanning tree in $G_\gamma^n$.


 
Clearly $T^n$ is a spanning tree in $G_\gamma^n$, since $T^n$ consists of tree edges identified by BFS in the connected graph $G_\gamma^n$.

We now show that the embedded tree  $T^n$ is a good spanning tree in the embedded graph $G_\gamma^n$. 
Firstly, each vertex $v$ of $T^n$ satisfies the condition (Cond1) of good spanning trees, because the tree edges are marked from BFS steps. 
According to the property of BFS, every edge of $G_\gamma^n$, whether a tree or a non-tree edge, joins two vertices whose levels differ by at most one. For a non-tree edge $e=(w_j, v)$, if $w_j$ lies on the path $P(r,v)$ in $T^n$ then $w_j$ is the parent of $v$. In this case  $e$  must be a tree edge. Similarly, $v$ does not lie on the path $P(r,w_j)$  in $T^n$   when $e=(w_j, v)$ is a non-tree edge. Hence, (Cond1) holds.

 We next show that $T^n$ satisfies (Cond2). Consider the situation when we dealt with edge $(x,y)$ while constructing $G_\gamma^j$ and $T^j$. 
  For the non-tree edge $(x, y)$, the non-tree edges in $Z_x$ are consecutive non-tree edges with respect to $x$ in $G_\gamma^n$. Let $Z$ be the set of vertices that contain other end vertices of the edges in $Z_x$  edges.  Clearly each vertex in $Z$    must be inside the cycle induced by $(x, y)$ in $T^n$. One can easily observe that if we traverse  $T^n$ as counterclockwise postorder traversal starting form the BFS-Start edge, the vertex $x$ will be visited after visiting all vertices in $Z$. Thus each vertex in $Z$ is contained in $T_{u_i^R}$ where $u_i$ is a vertex that lies on the path $P(r,x)$ and $x\ne u_i$ in  $T^n$. Similarly for the vertex $y$, it can be shown that the other end vertices of the edges in  $X_y$ are contained in  $T_{u_i^L}$. 
Thus one can easily observe that all edges incident to a vertex $v$ in $T^n$ can be partitioned into three consecutive edge sets $X_v$, $Y_v$ and $Z_v$. Hence, (Cond2) holds. Hence $T^n$ is a good spanning tree in $G_\phi=G_\gamma^n$.

   
  \end{proof}

 We have the following corollary based on the proof of Lemma~\ref{lemmatree}.
 \begin{corollary}
 \label{cccircl}
 Let $T$ be a good spanning tree in the embedding $G_\phi$ of $G$ obtained by the construction given in the proof of Lemma~\ref{lemmatree}. Then $T$ always has an edge with the property of the reference edge, mentioned in the proof of Lemma~\ref{theo1}.
\end{corollary}

  \begin{proof}

The DFS-Start edge can be taken as the reference edge, mentioned in the proof of Lemma~\ref{theo1}. 
  
  \end{proof}

The following theorem is the main result of this paper.
\begin{theorem}
Every connected planar graph of $n$ vertices admits a monotone grid drawing on a grid of area $O(n)\times O(n^2)$, and such a drawing can be found in $O(n)$  time.
 \label{finalTh}
\end{theorem}

  \begin{proof}
Let $G$ be a connected planar graph. By Lemma~\ref{lemmatree} $G$ has a plane embedding $G_\phi$ such that $G_\phi$ contains a good spanning tree $T$. By Lemma~\ref{theo1}  $G_\phi$ admits a monotone grid drawing on a grid of area $O(n)\times O(n^2)$. 

We can find a good spanning tree $G_\phi$ of  $G$ using the construction in the proof of Lemma~\ref{lemmatree}. After visiting each vertex during construction 
we need to identify $v$-components and  $\{u,v\}$-split components for a non-tree edge $(u,v)$  then we need to check whether these components are inside of  $G_j(C)$ in the intermediate step. If any component is found then we need to move out the component. 



A $v$-component is introduced by a cut vertex. All cut vertices can be found in $O(n+m)$ time using DFS. $v$-components and $\{u,v\}$-split components can also be found in linear time~\cite{TAR73}. We maintain a data structure to store each cut vertex or every pair of vertices with split components. We then use this record in the intermediate steps for finding $G_\phi$ in Lemma~\ref{lemmatree}. 
Let us assume we are traversing $(u,v)$ non-tree edge in a intermediate step $j$ of our algorithm. We can check whether any $u$-components, $v$-components  and $\{u,v\}$-split components of $G$ exist in $G_j(C)$ of  Lemma~\ref{lemmatree} by checking each edges incident to $u$ and $v$. This checking costs $O(d(u)+d(v))$  time. Throughout the algorithm it needs $O(m)$ time.
For moving a component outside of cycle $C$ we need to change at most four pointers in the adjacency list of $u$ and $v$, which takes $O(1)$ time. Hence the required time is $O(m)$. 
Since $G$ is a planar graph,  $G_\phi$ can be constructed in $O(n)$ time.

After constructing  $G_\phi$ we can construct a monotone drawing of $G$ using a recursive algorithm based on the inducting proof of Lemma~\ref{theo1}. It is not difficult to implement the recursive algorithm in  $O(n)$ time.

  \end{proof}


\section{Conclusion}
\label{conclusion}
In this paper we have studied monotone  grid drawings of planar graphs. We have shown that a connected planar graph of $n$ vertices has a straight-line planar monotone drawing on an $O(n)\times O(n^2)$ grid and we can find such a drawing in  $O(n)$  time. Finding straight-line monotone grid drawings of planar graphs on a smaller grid is our future work.

 \bibliographystyle{splncs03}
 \bibliography{main}




\end{document}